\documentclass[envcountsame,a4paper]{article}

\usepackage{amsmath, amsthm, amssymb}
\usepackage{url}
\usepackage{enumerate}
\usepackage{colonequals}

\theoremstyle{plain}
\newtheorem{theorem}{Theorem}
\newtheorem{lemma}[theorem]{Lemma}
\newtheorem{corollary}[theorem]{Corollary}

\theoremstyle{definition}
\newtheorem{definition}[theorem]{Definition}

\title{Justification logic enjoys the strong finite model property}
\author{Thomas Studer}
\date{}

\newcommand{\terms}{\text{Tm}}
\newcommand{\propositions}{\text{Prop}}
\newcommand{\justifies}{:}
\newcommand{\formulae}{\text{Fm}}

\newcommand{\model}{\mathcal{M}}
\newcommand{\worlds}{W}
\newcommand{\relation}{R}
\newcommand{\evidence}{\mathcal{E}}
\newcommand{\base}{\mathcal{B}}
\newcommand{\valuation}{\nu}

\newcommand{\forces}{\Vdash}

\newcommand{\pset}{\mathcal{P}}

\newcommand{\CS}{\mathsf{CS}}

\newcommand{\Low}{\mathsf{L}}
\newcommand{\Lo}{\mathsf{L}_\CS}
\newcommand{\J}{\mathsf{J}_\CS}
\newcommand{\JD}{\mathsf{JD}_\CS}
\newcommand{\JDF}{\mathsf{JD4}_\CS}
\newcommand{\JF}{\mathsf{J4}_\CS}
\newcommand{\JT}{\mathsf{JT}_\CS}
\newcommand{\LP}{\mathsf{LP}_\CS}

\newcommand{\ax}{\text{A}}
\newcommand{\axd}{\text{jd}}
\newcommand{\axt}{\text{jt}}
\newcommand{\axfour}{\text{j4}}

\newcommand{\size}[1]{|#1|}

\begin{document}
 \maketitle
 
 \begin{abstract}
We observe that justification logic enjoys a form the strong finite model property (sometimes also called small model property). Thus we obtain decidability proofs for justification logic that do not rely on Post's theorem.
\end{abstract}
 
 \section{Introduction}

Justification logics~\cite{ArtFit11SEP} are a family of logics that that, like modal logics, can express knowledge or provability of propositions.
However, instead of an implicit $\Box$-operator justification logics include explicit modalities of the form $t:$ where $t$ is a term representing a reason for an agent's knowledge or a proof a proposition.

Artemov developed the first justification logic~\cite{Art95TR,Art01BSL} to provide intuitionistic logic with a classical provability semantics.
Later Fitting~\cite{Fit05APAL} introduced epistemic models for justification logic. In this semantics, justification terms represent evidence a very general sense. For instance, our belief in $A$ may be justified by direct observation of $A$ or by learning that a friend heard about $A$. This general reading of justification led to a big variety of  epistemic justification logics for many different applications~\cite{Art08RSL,BucKuzStu11JANCL,BucKuzStu11WoLLIC,KMOS15,KuzStu12AiML}. 

There are many known decidability results for justification logics, see, for instance, \cite{BucKuzStu13LNCS,Kuz08PhD,Stu13JSL}.
However, many of these decidability proofs rely on completeness with respect to a recursively enumerable class of models and Post's theorem~\cite{pos44}.

In the present note we show that justification logic enjoys a form of the strong finite model property (which sometimes is called small model property)~\cite{brv02}. Thus we obtain decidability proofs for justification logics that do not make use of Post's theorem.

This note makes heavy use of~\cite{BucKuzStu11WoLLIC}.

\section{Justification Logics}
\label{sec:justificationlogics}

Justification terms are built from countably many constants $c_i$ and  countably many variables $x_i$ according to the following grammar:
\[
t \coloncolonequals c_i \ |\ x_i \ |\  (t \cdot t) \ |\ (t+t)   \ |\ !t \quad . 
\]
We denote the set of terms by $\terms$.
Formulae are built from  countably many atomic propositions~$p_i$ according to the following grammar:
\[
F \coloncolonequals p_i \ |\  \lnot F \ |\ (F \to F) \ |\ t \justifies F \quad. 
\]
$\propositions$ denotes the set of atomic propositions and $\formulae$ denotes the set of formulae.

The axioms of $\J$ consist of all instances of the following schemes:
\begin{description}
\item[\ax1] finitely many schemes axiomatizing  classical propositional logic
\item[\ax2] $t \justifies (A \rightarrow B) \rightarrow ( s \justifies A \rightarrow t\cdot s \justifies B)$
\item[\ax3] $t \justifies A \lor s \justifies A \rightarrow t+s \justifies A$
\end{description}

We will consider extension of $\J$ by the following axioms schemes.
\begin{description}
\item[(\axd)] $t \justifies \bot \to \bot$
\item[(\axt)] $t\justifies A \to A$
\item[(\axfour)] $t\justifies A \to !t \justifies t \justifies A$
\end{description}

A {\em constant specification} $\CS$ for a logic $\Low$ is any subset
\[
\CS \subseteq \{ c \justifies A  \ |\ \text{$c$ is a constant and $A$ is an axiom of $\Low$} \}.
\]

A constant specification $\CS$ for a logic $\Low$ is called
\begin{enumerate}
\item {\em axiomatically appropriate} if for each axiom $A$ of $\Low$ there is a constant $c$ such that $c \justifies A \in \CS$
\item {\em schematic} if for each constant $c$ the set $\{A \ |\ c \justifies A \in \CS\}$ consists of one or several (possibly zero)
axiom schemes, i.e.,~every constant justifies certain axiom schemes.
\end{enumerate}

For a constant specification $\CS$ the deductive system $\J$ is the Hilbert system given by the axioms \ax1--\ax3  and 
by the rules modus ponens and axiom necessitation:
  \[
  \begin{array}{c}
    A \quad A\rightarrow B
    \\ \hline
    B
  \end{array}\;\text{\scriptsize(MP)}\enspace,
  \qquad
  \begin{array}{c}
    c \justifies A \in \CS
    \\ \hline
    \underbrace{!! \cdots !}_n c \justifies \underbrace{! \cdots !}_{n-1} c \justifies \cdots \justifies !!c\justifies !c \justifies c \justifies A
  \end{array}\;\text{\scriptsize(AN!)}\enspace,
  \]
 where $n \geq 0$.
In the presence of the \axfour{} axiom a simplified axiom necessitation rule can be used:
  \[
  \begin{array}{c}
    c \justifies A \in \CS
    \\ \hline
    c \justifies A
  \end{array}\;\text{\scriptsize(AN)}\enspace.
  \]

Table \ref{table:systems} defines the various logics we  consider.
\begin{table}
\begin{center}
\begin{tabular}{l|ccccccccc}
 & \ax1 & \ax2 & \ax3 & \axd & \axt & \axfour & MP & AN! & AN \\ 
\hline
$\J$ & $\checkmark$ & $\checkmark$ & $\checkmark$ &   &    &   & $\checkmark$ & $\checkmark$ & \\
$\JD$ & $\checkmark$ & $\checkmark$ & $\checkmark$ & $\checkmark$  &    &   & $\checkmark$ & $\checkmark$ & \\
$\JT$ & $\checkmark$ & $\checkmark$ & $\checkmark$ &   &  $\checkmark$  &   & $\checkmark$ & $\checkmark$ & \\
$\JDF$ & $\checkmark$ & $\checkmark$ & $\checkmark$ & $\checkmark$  &    & $\checkmark$  & $\checkmark$ &  &$\checkmark$ \\
$\JF$ & $\checkmark$ & $\checkmark$ & $\checkmark$ &   &    & $\checkmark$  & $\checkmark$ & & $\checkmark$\\
$\LP$ & $\checkmark$ & $\checkmark$ & $\checkmark$ &   &  $\checkmark$  & $\checkmark$  & $\checkmark$ & & $\checkmark$  
\end{tabular}
\end{center}
\caption{Deductive Systems}
\label{table:systems}
\end{table}
We now present the semantics for these logics 

\begin{definition}[Evidence relation]\label{d:adevrel:1}
\newcounter{ic}
Let $(W,R)$ be a Kripke frame, i.e., $W\ne \varnothing$ and $R \subseteq W \times W$, and $\CS$ be a constant specification.
An admissible evidence relation $\evidence$ for a logic $\Lo$ is a subset of $\terms \times \formulae \times \worlds$ that satisfies the closure conditions:
\begin{enumerate}
\item if $(s,A,w) \in \evidence$ or $(t,A,w) \in \evidence$, then $(s+t,A,w) \in \evidence$
\item if $(s,A \to B,w) \in \evidence$ and $(t,A,w) \in \evidence$, then $(s \cdot t,B,w) \in \evidence$
\setcounter{ic}{\value{enumi}}
\end{enumerate}
Depending on whether or not the logic $\Lo$ contains the \axfour\ axiom,  the evidence function has to satisfy one
of the following two sets of closure conditions. If $\Lo$ does not include the \axfour\ axiom, then the additional requirement is:       
       \begin{enumerate}                                                                        
        \setcounter{enumi}{\value{ic}}
        \item if $c \justifies A \in \CS$ and $w \in \worlds$, then $(\underbrace{!!\cdots !}_{n} c,\underbrace{!\cdots !}_{n-1} c \justifies \cdots \justifies !!c \justifies !c \justifies c \justifies  A,w) \in \evidence$
        \setcounter{ic}{\value{enumi}}
       \end{enumerate}  
If $\Lo$ includes the \axfour\ axiom, then the additional requirement is:
       \begin{enumerate}
        \setcounter{enumi}{\value{ic}}
        \item if $c \justifies A \in \CS$ and $w \in W$, then $(c,A,w) \in \evidence$ 
        \item if $(t,A,w) \in \evidence$, then $(!t,t \justifies A,w) \in \evidence$
        \item \label{cond:mon:1} if $(t,A,w) \in \evidence$ and $w \relation v$, then $(t,A,v) \in \evidence$ 
       \end{enumerate}

\noindent If we drop condition \ref{cond:mon:1}, then we say $\evidence$ is a {\em t-evidence relation}.
Sometimes we use $\evidence(s,A,w)$ for $(s,A,w) \in \evidence$. 
\end{definition}

\begin{definition}[Evidence bases]\hfill
\label{def:evidencerelation}
\begin{enumerate}
\item An {\em evidence base} $\base$ is a subset of $\terms \times \formulae \times W$.
\item An evidence relation $\evidence$ is {\em based on} $\base$, if $\base \subseteq \evidence$.
\end{enumerate}
\end{definition} 

The closure conditions in the definition of admissible evidence function give rise to a monotone operator. The
minimal evidence relation based on $\base$ is the least fixed point of that operator and thus always exists.

\begin{definition}[Model]
Let $\CS$ be a constant specification. 
A \emph{Fitting model} for a logic $\Lo$ is a quadruple $\model = \left( \worlds, \relation, \evidence, \valuation \right)$ where
  \begin{itemize}
   \item $(\worlds, \relation)$ is a Kripke frame such that
   	\begin{itemize}
   	\item if $\Lo$ includes the \axfour\ axiom, then $\relation$ is transitive;
   	\item if $\Lo$ includes the \axt\ axiom, then $\relation$ is reflexive;
   	\item if $\Lo$ includes the \axd\ axiom, then $\relation$ is serial.
   	\end{itemize}
   \item $\evidence$ is an admissible evidence relation for $\Lo$ over the frame $(\worlds,\relation)$,
   \item $\valuation: \propositions \to \pset(\worlds)$, called  a valuation function.
  \end{itemize}
\end{definition}

\begin{definition}[Satisfaction relation]
The  relation of formula $A$ being satisfied in a model $\model = \left( \worlds, \relation, \evidence, \valuation \right)$ at a world $w \in \worlds$ is defined by induction on the structure of $A$ by
\begin{itemize}
 \item $\model, w \forces p_i$ if and only if $w \in \valuation(p_i)$
 \item $\forces$ commutes with Boolean connectives
 \item $\model, w \forces t \justifies B$ if and only if
  \begin{enumerate}[1)]
   \item $\model, v \forces B$ for all $v \in \worlds$ with $wRv$ and
   \item $(t,B,w) \in \evidence$
  \end{enumerate}
\end{itemize}
We say a formula $A$ is valid in a model  $\model = \left( \worlds, \relation, \evidence, \valuation \right)$ if 
for all $w \in \worlds$ we have $\model, w \forces A$. We say a formula $A$ is valid for a logic $\Lo$ if for all models $\model$ for $\Lo$ we have that $A$ is valid in $\model$.
\end{definition} 

The logics defined above are sound and complete (with a restriction in case of the logics containing the $\axd$ axiom). See~\cite{Art08RSL,Fit05APAL,Pac05PLS} for the full proofs of the following results.

\begin{theorem}[Soundness]
Let $\CS$ be a constant specification.
If a formula $A$ is derivable in a logic $\Lo$, then $A$ is valid for $\Lo$.
\end{theorem}

\begin{theorem}[Completeness]
\label{th:completeness}
\begin{enumerate}
\item Let $\CS$ be a  constant specification. If a formula $A$ 
is not derivable in $\Lo \in \{\J,\JT,\JF,\LP\}$, then there exists a model $\model$ for $\Lo$
with $\model, w \not\forces A$ for some world $w$ in $\model$.
\item Let $\CS$ be an axiomatically appropriate constant specification. If a formula $A$
is not derivable in $\Lo \in \{\JD,\JDF\}$, then there exists a model $\model$ for $\Lo$
with $\model, w \not\forces A$ for some world $w$ in $\model$.
\end{enumerate}
\end{theorem}

\section{The Strong Finite Model Property and Decidability}
\label{sec:decidability}

In this section we define and establish the strong finitary model property for many justification logics. As a corollary we get decidability proofs for these logics.

\begin{definition}[Finitary model]
\label{def:finitarymodel}
A model $\model = \left( \worlds, \relation, \evidence, \valuation \right)$ is called {\em finitary} if
\begin{enumerate}
\item $W$ is finite,
\item there exists a finite base $\base$ such that $\evidence$ is the minimal evidence relation based on $\base$, and
\item \label{def:finitarymodel:finiteevaluationfunction} the set $\{(w,p) \in \worlds \times \propositions \mid w\in \valuation(p)\}$ is finite.
\end{enumerate}
If $\model = \left( \worlds, \relation, \evidence, \valuation \right)$ is a finitary model for $\Lo$, then will sometimes specify this model by the tuple $\left( \worlds, \relation, \base, \valuation \right)$ where $\base$ is the finite base for $\evidence$.
\end{definition} 

Making use of filtrations for justification logics, we obtain the following theorem~\cite{BucKuzStu13LNCS}.

\begin{lemma}[Completeness w.r.t.\ finitary models]\label{l:finitarycompleteness}\hfill
\begin{enumerate}
\item Let $\Lo \in \{\J,\JT,\JF,\LP\}$ and $\CS$ be a constant specification for $\Low$. If a formula $A$ is 
not derivable in $\Lo$, then there exists a finitary model $\model$ for $\Lo$
with $\model, w \not\forces A$ for some world $w$ in $\model$.
\item Let $\Lo \in \{\JD,\JDF\}$ and $\CS$ be an axiomatically appropriate constant specification for $\Low$. If a formula $A$
is not derivable in $\Lo$, then there exists a finitary model $\model$ for $\Lo$
with $\model, w \not\forces A$ for some world $w$ in $\model$.
\end{enumerate}
\end{lemma}

\begin{definition}\hfill
\begin{enumerate}
\item Let $A$ be a formula. We denote the length of $A$ (i.e.~the number symbols in $A$) by $\size{A}$. 
\item Let $\Gamma$ be a set. We denote the cardinality of $\Gamma$ (i.e.~the number of elements of $\Gamma$) by  $\size{\Gamma}$. 
\end{enumerate}
\end{definition}

\begin{definition}[Strong finitary model property]\label{def:strong:1}
A justification logic $\Lo$ has the \emph{strong finitary model property}
if there are computable functions $f,g,h$ such that
for any formula $A$ that is not satisfiable, there exists a finitary model $\model = \left( \worlds, \relation, \base, \valuation \right)$ for $\Lo$ with
\begin{enumerate}
\item $\model, w \not\forces A$ for some $w \in \worlds$
\item $\size{\worlds} \leq f(\size{A})$,
\item $\size{\base} \leq g(\size{A})$,
\item $\size{\valuation} \leq h(\size{A})$.
\end{enumerate}
\end{definition}

Given the proof of Lemma~\ref{l:finitarycompleteness} in~\cite{BucKuzStu13LNCS} it is easy to see that we can effectively compute bounds on the size of the finitary model. Thus we get the strong finitary model property as a corollary of  Lemma~\ref{l:finitarycompleteness}.

\begin{corollary}[Strong finitary model property]\label{cor:strong:1}\hfill
\begin{enumerate}
\item Let $\Lo \in \{\J,\JT,\JF,\LP\}$ and $\CS$ be a constant specification for\/ $\Low$. 
Then $\Lo$ has the strong finitary model property.
\item Let $\Lo \in \{\JD,\JDF\}$ and $\CS$ be an axiomatically appropriate constant specification for\/ $\Low$.
Then $\Lo$ has the strong finitary model property.
\end{enumerate}
\end{corollary}

For a proof of the following lemma see~\cite[Lemma~4.4.6]{Kuz08PhD}.

\begin{lemma}
\label{lem:evidencedecidable}
Let $\CS$ be a decidable schematic constant specification and $\Lo \in \{\J,\JD,\JDF, \JT,\JF,\LP\}$.
Let $\model = (W, R, \evidence, \valuation)$ be a finitary model for $\Lo$.
Then the relation  $\model, w \forces A$ between worlds $w\in W$ and formulae~$A$ is decidable.
\end{lemma}

\begin{corollary}[Decidability]\hfill
\begin{enumerate}
\item Any justification logic in  $\{\J,\JT,\JF,\LP\}$ with a decidable schematic $\CS$ is decidable.
\item Any justification logic in $\{\JD,\JDF\}$ with a decidable, schematic and axiomatically appropriate $\CS$ is decidable.
\end{enumerate}
\end{corollary}
\begin{proof}
Let $\Lo$ be one of the above justification logics.
Given a formula $A$ we can generate all finitary models $\model = \left( \worlds, \relation, \base, \valuation \right)$ for $\Lo$ with
\begin{enumerate}
\item $\size{\worlds} \leq f(\size{A})$,
\item $\size{\base} \leq g(\size{A})$,
\item $\size{\valuation} \leq h(\size{A})$,
\end{enumerate}
for the functions $f, g, h$ from Definition~\ref{def:strong:1}. 
Note that we can decide whether a structure  $\model =   \left( \worlds, \relation, \base, \valuation \right)$ is a model for $\Lo$ since the required conditions on the accessibility relation, some combination of transitivity, reflexivity, and seriality can be effectively verified.

By Lemma~\ref{lem:evidencedecidable} we can decide for each of these finitary models, whether $\model, w \forces A$ for all $w \in \worlds$.

Making use of Corollary~\ref{cor:strong:1} we know that if $A$ is not $\Lo$-satisfiable, then the above procedure will generate a finitary model  $\model =   \left( \worlds, \relation, \base, \valuation \right)$ such that  $\model, w \not \forces A$ for some $w \in \worlds$.
Therefore, we conclude that satisfiability for $\Lo$ is decidable.
\end{proof}

\section{Conclusion}
\label{sec:conclusion}

We observed that justification logic enjoys a form of the strong finite model property (sometimes also called small model property). Thus we obtain decidability proofs for justification logics that do not rely on Post's theorem.


\end{document}